\documentclass[11pt,a 4paper]{article}

\usepackage[english]{babel}
\usepackage[utf8]{inputenc}
\usepackage[T1]{fontenc}
\usepackage{lmodern}
\usepackage{amssymb,amsmath,amsthm,amsfonts}
\usepackage[margin=2.7cm]{geometry}
\usepackage{enumitem}
\usepackage{color}
\usepackage{cite}
\usepackage{tikz}
\usetikzlibrary{shapes, positioning, graphs, graphs.standard}
\usepackage{hyperref}
\usepackage{indentfirst}
\usepackage{todonotes}

\marginparwidth=\dimexpr \marginparwidth + 0.6cm\relax

\tikzgraphsset{math nodes,nodes={circle, draw, inner sep=3}}
\tikzstyle{vertex}=[draw,circle,minimum size=7pt,inner sep=0pt]

\newlist{shortlist}{itemize}{1}
\setlist[shortlist]{label=\textbullet,nosep}

\newtheorem{thm}{Theorem}

\newtheorem{cor}[thm]{Corollary}

\newtheorem{lem}[thm]{Lemma}

\theoremstyle{definition}

\newtheorem*{conj*}{Conjecture}

\newcommand{\bigO}{{\mathcal{O}}}

\newcommand{\ie}{{\emph{i.e., }}}

\newcommand{\etal}{{\emph{et al.}}}

\newcommand{\Path}[1]{\ensuremath{{\sf P}_{#1}}}          
\newcommand{\K}[1]{\ensuremath{{\sf K}_{#1}}}             
\newcommand{\M}[1]{\ensuremath{{\sf M}_{#1}}}             
\newcommand{\extG}[3]{\ensuremath{{\sf E}_{#1,#2,#3}}}    

\newcommand{\dg}{\ensuremath{deg}}              
\newcommand{\dist}{\ensuremath{d}}              
\newcommand{\ecc}{\ensuremath{\epsilon}}        
\newcommand{\W}{\ensuremath{W}}                 

\newcommand{\eci}{\ensuremath{\xi^c}}           

\makeatletter

\title{Maximum Eccentric Connectivity Index\\for Graphs with Given
  Diameter}
\author{
  \large
  \textsc{
    Pierre Hauweele\textsuperscript{1},
    Alain Hertz\textsuperscript{2}%
    \footnote{Corresponding author : email  alain.hertz@gerad.ca;  tel. +1-514 340 6053.},
    Hadrien M\'elot\textsuperscript{1},}\\
  \textsc{
    Bernard Ries\textsuperscript{3},
    Gauvain Devillez\textsuperscript{1}}\\[3mm]
  \normalsize \textsuperscript{1} Computer Science Department - Algorithms Lab\\
  \normalsize University of Mons, Mons, Belgium\\[3mm]
  \normalsize \textsuperscript{2} Department of Mathematics and Industrial
    Engineering\\
  \normalsize \'Ecole Polytechnique - Gerad, Montréal, Canada\\[3mm]
  \normalsize \textsuperscript{3} Department of Informatics\\
  \normalsize University of Fribourg, Fribourg, Switzerland
}

\begin{document}

\maketitle
\vspace*{0.2cm}

\begin{abstract}
  The eccentricity of a vertex $v$ in a graph $G$ is the maximum distance
  between $v$ and any other vertex of $G$. The diameter of a graph $G$ is
  the maximum eccentricity of a vertex in $G$. The eccentric connectivity
  index of a connected graph is the sum over all vertices of the product
  between  eccentricity and  degree. Given two integers $n$ and $D$
  with $D\leq n-1$, we characterize those graphs which have the largest
  eccentric connectivity index among all connected graphs of order $n$ and
  diameter $D$. As a corollary, we also characterize those graphs which
  have the largest eccentric connectivity index among all connected graphs
  of a given order $n$.
\end{abstract}
\section{Introduction}
Let $G=(V,E)$ be a simple connected undirected graph. The \emph{distance} $\dist(u,v)$
between two vertices $u$ and $v$ in $G$ is the number of edges of a shortest path in $G$ connecting $u$ and $v$. 
The \emph{eccentricity}
$\ecc(v)$ of a vertex $v$ is the maximum distance between $v$ and
any other vertex, that is $\max \{ \dist(v, w) ~|~ w \in V \}$. The
\emph{diameter} 
of $G$ is the maximum eccentricity among all
vertices of $G$. 
The \emph{eccentric connectivity index} $\eci(G)$ of $G$ is defined by
\[
\eci(G) = \sum_{v \in V} \dg(v) \ecc(v).
\]
This index was introduced by Sharma \etal~in~\cite{Sharma97}. Alternatively,
$\eci$ can be computed by summing the eccentricities of the extremities
of each edge:
\[
\eci(G) = \sum_{vw \in E} \ecc(v) + \ecc(w).
\]
We define the weight of a vertex by $\W(v) = \dg(v) \ecc(v)$, and we thus have $\eci(G) = \sum_{v \in V} \W(v)$.
Morgan \etal~\cite{Morgan2011} give the following asymptotic
upper bound on $\eci(G)$ for a graph $G$ of order $n$ and with a given diameter $D$.

\begin{thm}[Morgan, Mukwembi and Swart, 2011~\cite{Morgan2011}]
  Let $G$ be a connected graph of order $n$ and diameter $D$. Then,
  \[
    \eci(G) \le D(n-D)^2 + \bigO(n^2).
  \]
\end{thm}
In what follows, we write $G \simeq H$ if $G$ and $H$ are two isomorphic
graphs, and we let $\K{n}$ and $\Path{n}$ be the \emph{complete graph} and the
\emph{path} of order $n$, respectively.  We refer to Diestel~\cite{Diestel00}
for basic notions of graph theory that are not defined here.
A \emph{lollipop} $L_{n,D}$ is a graph obtained from a path $\Path{D}$ by
joining an end vertex of this path to $\K{n-D}$. Morgan \etal~\cite{Morgan2011}
state that the above asymptotic bound is best possible by showing that
$\eci(L_{n,D}) = D(n-D)^2 + \bigO(n^2)$. The aim of this paper is to give a
precise upper bound on $\eci(G)$ in terms of $n$ and $D$, and to completely
characterize those graphs that attain the bound. As a result, we will observe
that there are graphs $G$ of order $n$ and diameter $D$ such that $\eci(G)$
is strictly larger than $\eci(L_{n,D})$.

Morgan \etal~\cite{Morgan2011} also give an asymptotic upper bound on $\eci(G)$ for graphs $G$ of order $n$ (but without a fixed diameter), and show that this bound is sharp by observing that it is attained by $L_{n,\frac{n}{3}}$.

\begin{thm}[Morgan, Mukwembi and Swart, 2011~\cite{Morgan2011}]
  Let $G$ be a connected graph of order $n$. Then,
  \[
  \eci(G) \le \frac{4}{27}n^3 + \bigO(n^2).
  \]
\end{thm}
We give a precise upper bound on $\eci(G)$ for graphs $G$ of order $n$,
and characterize those graphs that reach the bound. As a corollary, we show
that for every lollipop, there is another graph $G$ of same order, but with
a strictly larger eccentric connectivity index.

\section{Results for a fixed order and a fixed diameter}

The only graph with diameter $1$ is the clique, and clearly,
$\eci(\K{n})=n(n-1)$. Also, the only connected graph with $3$ vertices and
diameter 2 is $\Path{3}$, and $\eci(\Path{3})=\eci(\K{3})=6$.  The next
theorem characterizes the graphs with maximum eccentric connectivity index
among those with $n\geq 4$ vertices and diameter $2$. Let  $\M{n}$ be the graph obtained from $\K{n}$ by removing
a maximum matching (\ie $\lfloor\frac{n}{2}\rfloor$ disjoint edges) and,
if $n$ is odd, an additional edge adjacent to the unique vertex that still
has degree $n-1$. In other words, all vertices in $\M{n}$ have degree $n-2$,
except possibly one that has degree $n-3$. For illustration, $\M{6}$ and $\M{7}$ are drawn in Figure \ref{fig:specialcases}.
\begin{figure}[!h]
  \begin{center}
  \begin{tikzpicture}
    \begin{scope}
    \graph[empty nodes] {
      subgraph C_n [n=4, clockwise] -- mid;
    };
    \node at (0, -1.6) {\(H_1\)};
    \end{scope}
    \begin{scope}[xshift={3cm}]
    \graph[empty nodes, no placement] {
      {subgraph K_n [n=4, clockwise, phase=45,
              nodes={xshift={1.5cm}}]};
      a -- {3, 4};
      {1, 2} -- b [x=3];
    };
    \node at (1.5, -1.6) {\(H_2\)};
    \end{scope}
    \begin{scope}[xshift={7.5cm}]
    \graph[empty nodes] {
      a -- {[nodes={yshift=1cm}] 1, 2, 3} --[complete bipartite]
      {[nodes={yshift=0.5cm}] 4, 5} -- b;
      4 -- 5;
      1 -- 2 -- 3;
      1 --[bend right] 3
    };
    \node at (1.5, -1.6) {\(H_3\)};
    \end{scope}
    \begin{scope}[yshift={-3.7cm}]
    \begin{scope}
    \graph[simple, empty nodes] {
      subgraph K_n [n=6, clockwise, phase=120],
        1 -!- 2, 3 -!- 4, 5 -!- 6;
     };
     \node at (0, -2) {\M{6}};
    \end{scope}
    \begin{scope}[xshift={4.5cm}]
    \graph[simple, empty nodes] {
      subgraph K_n [n=7, clockwise],
      1 -!- 2, 3 -!- 4, 5 -!- 6, 7 -!- 1;
    };
    \node at (0, -2) {\M{7}};
    \end{scope}
    \begin{scope}[xshift={7cm},yshift={-1cm},new set=K]
    \graph[empty nodes]{
      subgraph P_n[n=5,branch right,name=P],
      {[nodes={set=K}, no placement, clique]
        1[y=1], 2[x=1,y=1], 3[x=0.5, y=1.87cm]},
      {P 1, P 2} --[complete bipartite] (K),
      P 3 --[dashed] (K);
    };
    \foreach \v [count=\i from 0] in {1,...,5}{
      \node[below=.01cm of P \v] {\(u_\i\)};
    }
    \node at (2, -1) {\extG{8}{4}{k}};
    \end{scope}
    \end{scope}
  \end{tikzpicture}
  \end{center}
\vspace{-0.5cm}\caption{Graphs $H_1, H_2, H_3$, $\M{6}$, $\M{7}$ and $\extG{8}{4}{k}$, (dashed edges depend
  on $k$)}
  \label{fig:specialcases}
\end{figure}
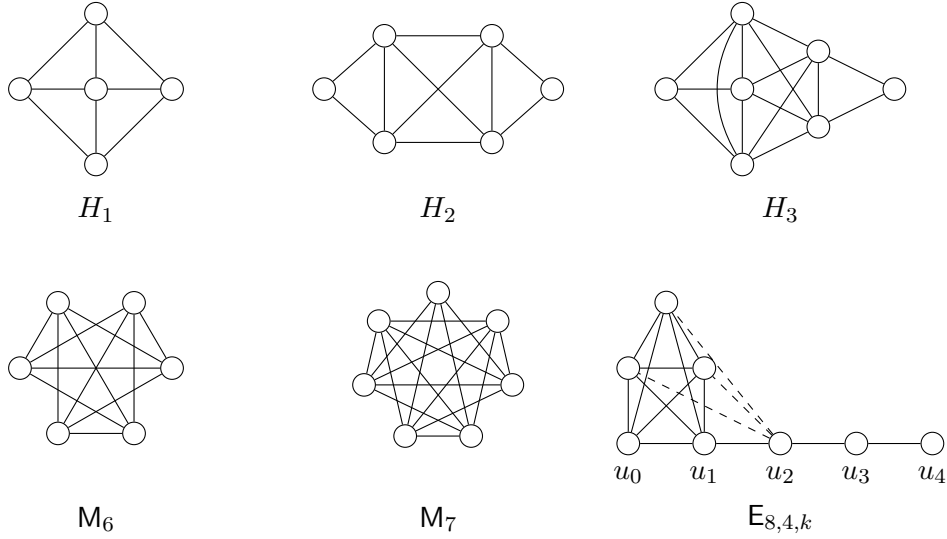

\begin{thm}
  Let $G$ be a connected graph of order $n\geq 4$ and diameter $2$. Then,
  \[
  \eci(G) \le 2n^2-4n-2(n\bmod 2)
  \]
  with equality if and only if $G\simeq\M{n}$ or $n=5$ and
  $G\simeq H_1$ (see Figure~\ref{fig:specialcases}).
\end{thm}
\begin{proof}
  Let $G$ be a graph of order $n$ and diameter $2$, and let $x$ be the number of vertices of degree $n-1$ in $G$. Clearly, $W(v)=n-1$ for all  vertices $v$ of degree $n-1$, while $W(v)\leq 2(n-2)$ for all other vertices $v$. Note that if $n-x$ is odd, then at least one vertex in $G$ has degree at most $n-3$. Hence,
\vspace{-0.2cm}  \begin{align*}
  \eci(G)&\leq x(n-1)+2(n-x)(n-2)-2((n-x)\bmod 2)\\
  &= 2n^2-4n+x(3-n)-2((n-x)\bmod 2).
  \end{align*}
  For $n=4$ or $n\geq 6$, this value is maximized with $x=0$. For $n=5$, both $x=1$ (\ie $G\simeq H_1$) and $x=0$ (\ie $G\simeq \M{5}$) give the maximum value $28=2n^2-4n+(3-n)-2((n-1) \bmod 2)=2n^2-4n-2(n \bmod 2)$.
\end{proof}

Before giving a similar result for graphs with diameter $D\geq 3$, we prove the following useful property.

\begin{lem} \label{lem:outsiders_connections_to_P}
  Let \(G \) be a connected graph of order \(n \geq 4\) and diameter
  $D\geq 3$. Let \(P\) be a shortest path in \(G\)  between two vertices at distance $D$, and assume there is a vertex \(u\) on \(P\) such that $\ecc(u)$ is strictly larger than the longest distance $L$ from \(u\) to an extremity of \(P\). Finally, let $v$ be a vertex in $G$ such that $d(v,u)=\ecc(u)$ and let $v=w_1 - w_2 -\ldots - w_{\ecc(u)+1}=u$ be a path of length $\ecc(u)$ linking $v$ to $u$ in $G$. Then
  \begin{shortlist}
    \item vertices $w_1, \ldots, w_{\ecc(u)-L}$ do not belong to \(P\);
    \item vertex $w_{\ecc(u)-L}$ has either no neighbor on \(P\), or its
    unique neighbor on \(P\) is an extremity at distance $L$ from $u$;
    \item if $\ecc(u)-L>1$ then vertices $w_1, \ldots, w_{\ecc(u)-L-1}$ have no neighbor on \(P\).
  \end{shortlist}
\end{lem}
\begin{proof}
  No vertex $w_i$ with $1\leq i \leq \ecc(u)-L$ is on $P$, since this would imply $d(u,w_i)\leq L$, and hence $d(u,v)=d(u,w_1)\leq L+i-1\leq \ecc(u)-1$. Similarly, no vertex $w_i$ with $1\leq i \leq \ecc(u)-L-1$ has a neighbor on $P$, since this would imply $d(u,w_i)\leq L+1$, and hence $d(u,v)=d(u,w_1)\leq L+1+i-1\leq \ecc(u)-1$. If vertex $w_{\ecc(u)-L}$ has at least one neighbor on $P$, then this neighbor is necessarily an extremity of $P$ at distance $L$ from $u$, else we would have $d(u,w_{\ecc(u)-L})\leq L$, which would imply $d(u,v)=d(u,w_1)\leq L+(\ecc(u)-L-1)=\ecc(u)-1$. We conclude the proof by observing that if both extremities of $P$ are at distance $L$ from $u$, then  $w_{\ecc(u)-L}$ is adjacent to at most one of them since $D\geq 3$.
	\end{proof}
	

Let $n, D$ and $k$ be integers such that $n \ge 4$, $3 \le D
\le n-1$ and $0 \le k \le n-D-1$, and let $\extG{n}{D}{k}$ be
the graph (of order $n$ and diameter $D$) constructed from a path
$u_0-u_1-\ldots-u_D$ by joining each vertex of a clique $\K{n-D-1}$
to $u_0$ and $u_1$, and $k$ vertices of the clique to $u_2$ (see
Figure~\ref{fig:specialcases}). Observe that $\extG{n}{D}{0}$ is the
lollipop $L_{n, D}$ and that $\extG{n}{D}{n-D-1}$ can be viewed as a
lollipop with a missing edge between $u_0$ and $u_2$. Also, if $D = n-1$,
then $k = 0$ and $\extG{n}{n-1}{0} \simeq \Path{n}$.

\begin{figure}[!h]
  \label{fig:extr_graphs}
\end{figure}

\begin{lem}
  \label{lem:eci_for_extremal_graph}
  Let $n, D$ and $k$ be integers such that $n \ge 4$, $3 \le D \le n-1$
  and $0 \le k \le n-D-1$, then
  \begin{align*}
    \eci(\extG{n}{D}{k}) ={}&
      2\sum_{i=0}^{D-1}\max\{i,D-i\}+\Big(n-D-1\Big)\Big( 2D-1+D(n-D)\Big) \\
    & +k \Big(2D-n-1+\max\{2,D-2\}\Big).
  \end{align*}
\end{lem}

\begin{proof}
The sum of the weights of the vertices outside $P$ is
\begin{align*}
  \sum_{v \in V \setminus V(P)} \W(v) &= k \left(
    n-D+1\right)\left(D-1\right) + (n-D-1-k) \left(n-D\right) D,\\
  &= k \left( 2D-n-1 \right) + (n-D-1)(n-D)D.\\
\end{align*}
We now consider the weights of the vertices in $P$. The weight of $u_0$ is $D(n-D)$, the weight of $u_1$ is $(D-1)(n-D+1)$, and the weight of $u_2$ is 
$(k+2)\max\{2,D-2\}$. The weight of $u_i$ for $i=3,\ldots,D-1$ is $2\max\{i,D-i\}$, and the weight of $u_D$ is $D$.
Hence, the total weight of the vertices on $P$ is
\begin{align*}
   & (n-D)D+(n-D+1)(D-1)+(k+2)\max\{2,D-2\}+2\sum_{i=3}^{D-1}\max\{i,D-i\}+D\\
={}& \Big((n-D-1)D+D\Big)+\Big((n-D-1)(D-1)+2(D-1)\Big)\\
   & +\Big(k\max\{2,D-2\}+2\max\{2,D-2\}\Big)+2\sum_{i=3}^{D-1}\max\{i,D-i\}+D\\
={}& 2\sum_{i=0}^{D-1}\max\{i,D-i\}+(n-D-1)(2D-1)+k\max\{2,D-2\}
\end{align*}
By summing up all weight in $G$, we obtain the desired result.
\end{proof}

In what follows, we denote $\displaystyle f(n,D)=\max_{k=0}^{n-D-1}\eci(\extG{n}{D}{k})$. It follows from the above lemma that
\[
f(n,D)=
\begin{cases}
14+\Big(n-4\Big)\Big(3n-4+\max\{0,2D-n+1\}\Big)&\text{if }D=3;\\[3.5ex]
\begin{split}
&2\sum_{i=0}^{D-1}\max\{i,D-i\}\\
&\quad+\Big(n-D-1\Big)\Big( 2D-1+D(n-D)+\max\{0,3D-n-3\}\Big)
\end{split} &
\text{if }D\geq 4.
\end{cases}
\]
Lemma~\ref{lem:eci_for_extremal_graph} allows to know for which values of $k$
we have $ \eci(\extG{n}{D}{k})=f(n,D)$.

\begin{cor}
  Let $n$ and $k$ be integers such that $n \ge 4$,
  and $0 \le k \le n-4$. 
  \begin{shortlist}
    \item If $n < 7$, then $\eci(\extG{n}{3}{k}) \le
      f(n,3)=2n^2\!-\!5n+2$, with equality if and only if  $k = n-4$.
    \item If $n > 7$, then $\eci(\extG{n}{3}{k}) \le
      f(n,3)=3n^2-16n+30$ with equality if and only if $k=0$.
    \item If $n = 7$, then all $\eci(\extG{n}{3}{k})$
      are equal to $65$ for $k = 0, \ldots, n-4$.
  \end{shortlist}
\end{cor}

\begin{cor} \label{cor:max_k}
  Let $n, D$ and $k$ be integers such that $n \ge 5$, $4 \le D \le n-1$ 
  and $0 \le k \le n-D-1$.
  \begin{shortlist}
    \item If $n < 3(D-1)$, then $\eci(\extG{n}{D}{k})=f(n,D)$ if and only
      if  $k = n-D-1$.
    \item If $n > 3(D-1)$, then $\eci(\extG{n}{D}{k})=f(n,D)$ if and only
      if $k=0$.
    \item If $n = 3(D-1)$, then $\eci(\extG{n}{D}{k})=f(n,D)$ if and only
      if $k\in\{ 0, \ldots, n-D-1\}$.
  \end{shortlist}
\end{cor}

The graph $H_2$ of Figure \ref{fig:specialcases} has $6$ vertices, diameter $3$, and is not isomorphic to $\extG{6}{3}{k}$, while  $\eci(H_2)=f(6,3)=44$. Similarly, the graph $H_3$ of Figure \ref{fig:specialcases} has $7$ vertices, diameter $3$, and is not isomorphic to $\extG{7}{3}{k}$, while  $\eci(H_3)=f(7,3)=65$. In what follows, we prove that all graphs $G$ of order $n$ and diameter $D\geq 3$ have $\eci(G)\leq f(n,D)$. Moreover, we show that if $G$ is not isomorphic to a $\extG{n}{D}{k}$, then equality can only occur if $G\simeq H_2$ or $G \simeq H_3$. So, for every $n\geq 4$ and $3\leq D\leq n-1$, let us consider the following graph class $\mathcal{C}_n^D$:
\[
\mathcal{C}_n^D=
\begin{cases}
  \{\extG{n}{3}{n-4}\} & \text{if } n=4,5  \text{ and } D=3 ;\\
  \{\extG{n}{3}{2},H_2\} & \text{if } n=6  \text{ and } D=3 ;\\
  \{\extG{n}{3}{0},\ldots,\extG{n}{3}{3},H_3\} & \text{if } n=7  \text{ and }  D=3 ;\\
  \{\extG{n}{3}{0}\} & \text{if } n>7  \text{ and }  D=3 ;\\
  \{\extG{n}{D}{n-D-1}\} & \text{if } n<3(D-1) \text{ and }  D\geq 4 ;\\
  \{\extG{n}{D}{0},\ldots,\extG{n}{D}{n-D-1}\} & \text{if } n=3(D-1)  \text{ and }  D\geq 4 ;\\
  \{\extG{n}{D}{0}\}  & \text{if } n > 3(D-1)  \text{ and }  D\geq 4 .
\end{cases}
\]
Note that while Morgan \etal~\cite{Morgan2011} state that the lollipops
reach the asymptotic upper bound of the eccentric connectivity index, we
will prove that they reach the more precise upper bound only if $D=n-1$,
$D=3$ and $n\geq 7$, or $D\geq 4$ and $n\geq 3(D-1)$.
\begin{thm} \label{thm:ub_diam}
  Let \(G\) be a connected graph of order \(n \ge 4\) and diameter
  \(3 \le D \le n-1\). Then $\eci(G) \le f(n,D)$, with equality if and only if $G$ belongs to $\mathcal{C}_n^D$.
\end{thm}

\begin{proof}
We have already observed that all graphs $G$ in $\mathcal{C}_n^D$ have $\eci(G)= f(n,D)$. So let $G$ be a graph of order $n$, diameter $D$ such that $\eci(G)\geq f(n,D)$. It remains to prove that $G$ belongs to $\mathcal{C}_n^D$.

Let \(P=u_0 - u_1 - \dots - u_D\) be a shortest path in \(G\) that connects two vertices $u_0$ and $u_D$ at distance $D$ from each other. In what follows, we use the following notations for all $i=0,\ldots,D$:
\begin{shortlist}
	\item \(o_i\) is the number of vertices outside $P$ and adjacent to $u_i$;
	\item \(\delta_i = \max\{i, D-i\}\);
	\item $r_i=\ecc(u_i)-\delta_i$.
\end{shortlist}
Also, let $r^*=\max_{i=1}^{D-1}r_i$. Note that $\delta_i\geq 2$ and $r_i\leq \lfloor\frac{D}{2}\rfloor$ for all $i$, and $r_0=r_D=0$ since $\ecc(u_0)=\ecc(u_D)=\delta_0=\delta_D=D$. Since $P$ is a shortest path linking $u_0$ to $u_D$, no vertex  outside $P$ can have more than three neighbors in $P$. We consider the following partition of the vertices outside $P$ in 4 disjoint sets $V_0,V_{1,2},V_3^{D-1},V_3^D$, and denote by $n_0,n_{1,2},n_3^{D-1},n_3^D$ their respective size:
\begin{shortlist}
	\item $V_0$ is the set of vertices outside $P$ with no neighbor on $P$;
	\item $V_{1,2}$ is the set of vertices outside $P$ with one or two neighbors in $P$;
	\item $V_3^{D-1}$ is the set of vertices $v$ outside $P$ with three neighbors in $P$ and $\ecc(v)\leq D-1$;
	\item $V_3^{D}$ is the set of vertices $v$ outside $P$ with three neighbors in $P$ and $\ecc(v)=D$.
\end{shortlist}

\vspace{0.2cm}Clearly, all vertices $v$ outside $P$ can have $\ecc(v)=D$ except those in $V_3^{D-1}$. The maximum degree of a vertex in $V_0$ is $n-D-2$, while it is $n-D$ for those in $V_{1,2}$ and $n-D+1$ for those in $V_3^{D-1}\cup V_3^{D}$.
For a vertex $v\in V_{1,2} \cup V_3^{D-1} \cup V_3^{D}$, let 
\begin{gather*}
  \rho(v)=\max\{r_i \;|\; u_i \text{ is adjacent to }v\}\\
  \rho^*=\max_{v\in V_{1,2} \cup V_3^{D-1} \cup V_3^{D}}\rho(v)
\end{gather*}
Hence, $r^*\geq \rho^*$.
We first show that the total weight of the vertices in $V_0\cup V_{1,2}$ is at most 
\begin{align*}
     &D(n-D)(n-D-1-n_3^{D-1}-n_3^{D})-2Dr^*+D\min\{1, \rho^*\}.
\end{align*}
\begin{itemize}
	\item If $r^*=0$, then the largest possible weight of the vertices in $V_0\cup V_{1,2}$ occurs when all of them have two neighbors in $P$ (i.e., $n_0=0$ and no vertex in $V_{1,2}$ has one neighbor on $P$). In such a case, $n_0+n_{1,2}=n-D-1-n_3^{D-1}-n_3^{D}$, and all these vertices have degree $n-D$. Hence, their total weight is at most $D(n-D)(n-D-1-n_3^{D-1}-n_3^{D})$.
	\item If $r^*>0$ and $\rho^*>0$, then let $i$ be such that $r_i=r^*$. It follows from Lemma \ref{lem:outsiders_connections_to_P} that there is a path $w_1 -\ldots - w_{\ecc(u_i)+1}$ such that $w_1,\ldots,w_{r^*-1}$ have no neighbor on $P$ and $w_{r^*}$ has at most one neighbor on $P$. Hence, the largest possible weight of the vertices in $V_0\cup V_{1,2}$ occurs when $r^*-1$ vertices have 0 neighbor on $P$, one vertex has one neighbor on $P$, and $n-D-1-n_3^{D-1}-n_3^{D}-r^*$ vertices have 2 neighbors in $P$. Hence, the largest possible weight for the vertices in $V_0\cup V_{1,2}$ is
	\begin{align*}
	   & D(n-D-2)(r^*-1)+D(n-D-1)+D(n-D)(n-D-1-n_3^{D-1}-n_3^{D}-r^*)\\
	={}& D(n-D)(n-D-1-n_3^{D-1}-n_3^{D})-2Dr^*+D.
	\end{align*}
	\item If $r^*>0$ and $\rho^*=0$, then consider the same path $w_1 -\ldots - w_{\ecc(u_i)+1}$ as in the above case. If $w_{r^*}$ has no neighbor on $P$, then there are at least $r^*$ vertices with no neighbor on $P$ and the largest possible weight for the vertices in $V_0\cup V_{1,2}$ is
	\[
	\begin{array}{l}
	D(n-D-2)(r^*)+D(n-D)(n-D-1-n_3^{D-1}-n_3^{D}-r^*)\\
	=D(n-D)(n-D-1-n_3^{D-1}-n_3^{D})-2Dr^*.
	\end{array}
	\]
	Also, if there are at least two vertices in $V_{1,2}$ with only one neighbor on $P$, then the largest possible weight for the vertices in $V_0\cup V_{1,2}$ is
	\[
	\begin{array}{l}
	D(n-D-2)(r^*-1)+2D(n-D-1)+D(n-D)(n-D-1-n_3^{D-1}-n_3^{D}-r^*-1)\\
	=D(n-D)(n-D-1-n_3^{D-1}-n_3^{D})-2Dr^*.
	\end{array}
	\]
	So assume $w_{r^*}$ is the only vertex in $V_{1,2}$ with only one neighbor on $P$. We thus have $\dist(u_i,w_{r^*})=\delta_i+1$. We now show that this case is impossible. We know from Lemma \ref{lem:outsiders_connections_to_P} that $w_{r^*}$ is adjacent to $u_0$ or (exclusive) to $u_D$. Since $\rho(v)=0$ for all vertices $v$ outside $P$, we know that $u_i$ has no neighbor outside $P$. Hence, $w_{\ecc(u_i)}$ is $u_{i-1}$ or $u_{i+1}$, say $u_{i+1}$ (the other case is similar). Then $w_{r^*}$ is not adjacent to $u_0$ else there is $j$ with $r^*+1\leq j \leq \ecc(u_i)-1$ such that $w_j$ is outside $P$ and has $w_{j+1}$ as neighbor on $P$, and since $w_j$ must have a second neighbor $u_{\ell}$ on $P$ with $\ell\geq i+2$, we would have  
	\[
	i+2\leq \ell=\dist(u_0,u_{\ell})\leq \dist(w_{r^*},w_j)+2 \leq (\dist(w_{r^*},u_i)-2)+2=i+1.
	\]
	
	Hence, $w_{r^*}$ is adjacent to $u_D$. Then there is also a path linking $u_i$ to $w_1$ going through $u_{i-1}$ else $d(u_0,w_1) = d(u_0,u_i)+d(u_i,w_1)>i+\delta_i\geq D$. Let $Q$ be such a path of minimum length. Clearly, $Q$ has length at least equal to $\ecc(u_i)$. So let $w'_1 - \ldots - w'_{\ecc(u_i)+1}$ be the subpath of $Q$ of length $\ecc(u_i)$ and having $u_i$ as extremity (i.e., $w'_{\ecc(u_i)}=u_{i-1}$ and $w'_{\ecc(u_i)+1}=u_{i}$). Applying the same argument to $w'_{r^*}$ as was done for $w_{r^*}$, we conclude that $w'_{r^*}$ has $u_0$ as unique neighbor on $P$. We thus have two vertices in  $V_{1,2}$ with a unique neighbor on $P$, a contradiction.
\end{itemize}

The total weight of the vertices in $V_3^{D-1}\cup V_3^{D}$ is at most 
$(n-D+1)\Big((D-1)n_3^{D-1}+Dn_3^{D}\Big)$, which gives the following upper bound $B$ on the total weight of the vertices outside $P$:
\begin{align*}
  B={}& D(n-D)(n-D-1-n_3^{D-1}-n_3^{D})+(n-D+1)\Big((D-1)n_3^{D-1}+Dn_3^{D}\Big)\\
      &-2Dr^*+D\min\{1,\rho^*\}\\
   ={}&(n-D-1)D(n-D)+n_3^{D-1}(2D-n-1)+Dn_3^{D}-2Dr^*+D\min\{1,\rho^*\}.
\end{align*}
This bound can only be reached if all vertices outside $P$ are pairwise adjacent. But Lemma~\ref{lem:outsiders_connections_to_P} shows that this cannot happen if $\rho^*>0$. Indeed, consider a vertex $v$ in $V_{1,2}\cup V_{3}^{D}\cup V_{3}^{D-1}$ with $\rho(v)>0$. There is  a vertex $u_i$ in $P$ adjacent to $v$ such  that $\rho(v)=r_i=\ecc(u_i)-\delta_i>0$. We know from Lemma \ref{lem:outsiders_connections_to_P} that there is a shortest path $w_1 - w_2 -\ldots - w_{\ecc(u_i)+1}=u_i$ linking $u_i$ to a vertex $w_1$ with $d(u_i,w_1)=\ecc(u_i)$ and such that $w_1, \ldots, w_{\rho(v)}$ do not belong to \(P\). In what follows, we denote $Q^v$ such a path. If $v$ is adjacent to a $w_j$ with $1\leq j\leq \rho(v)$, then the path $u_i - v - w_j - \ldots - w_1$ links $u_i$ to $w_1$ and has length at most $\rho(v)+1<r_i+\delta_i=\ecc(u_i)$, a contradiction. Hence $v$ has at least $\rho(v)$ non-neighbors outside $P$. Also, as shown in Lemma \ref{lem:outsiders_connections_to_P}, $w_1,\ldots, w_{\rho(v)-1}$ belong to $V_0$, while $w_{\rho(v)}$ belongs to $V_0\cup V_{1,2}$. In the upper bound $B$, we have assumed that $\ecc(w_1)=\ldots=\ecc(w_{\rho(v)})=D$. Hence, if $v\in V_{1,2}\cup V_{3}^{D}$, we can gain $2D$ units on $B$ for every $w_j$, $j=1,\ldots,\rho(v)$ ($D$ for $v$ and $D$ for $w_j$), while the gain is $2D-1$ ($D-1$ for $v$ and $D$ for $w_j$) if $v\in V_{3}^{D-1}$. 

We can gain an additional $2D$ for every $v\in V_{3}^{D}$. Indeed, consider such a vertex $v$ and let $w^*$ be a vertex at distance $D$ from $v$. Note that $w^*$ is not on $P$ and has at most one neighbor on $P$ else $d(v,w^*)\leq D-1$. Hence, if $\rho(v)=0$, we can gain $2D$ (one $D$ for $v$ and one $D$ for $w$) in the above upper bound. So assume $\rho(v)>0$, and consider again the shortest path $Q^v=w_1 - w_2 -\ldots - w_{\ecc(u_i)+1}=u_i$, with $\rho(v)=r_i$. Also, let $W=\{w_1,\ldots,w_{\rho(v)}\}$. To gain an additional $2D$, it is sufficient to determine a vertex in $(V_0\cup V_{1,2})\setminus W$ which is not adjacent to $v$. So assume no such vertex exists, and let us prove that such a situation cannot occur. Note that $w^*\notin V_{3}^{D}\cup V_{3}^{D-1}$ (since it has at most one neighbor on $P$), which implies
$w^*\in W$.

	\begin{itemize}
		\item If a vertex $w_j\in W$ has a neighbor $x\in V_0\cup V_{1,2}$ outside $W$, then $v$ is adjacent to $x$, and the path $v - x - w_j - \ldots - w^*$ has length at most $1+\rho(v)\leq 1+\lfloor\frac{D}{2}\rfloor<D$, a contradiction. 		
		\item If a vertex $w_j\in W$ has a neighbor $x\in V_3^D\cup V_3^{D-1}$, then $d(u_i,w_1)\leq d(u_i,x)+d(x,w_1)\leq \delta_i-1+r_i<\ecc(u_i)$, a contradiction. 		
	\end{itemize}
	Since $G$ is connected and $w_1,\ldots,w_{\rho(v)-1}$ have no neighbors outside $Q^v$, we know that $w_{\rho(v)}$ is adjacent to the extremity of $P$ at distance $\delta_i$ from $u_i$ (and to no other vertex on $P$). Hence, the vertices on $P$ and those in $W$ induce a path of length $D+\rho(v)>D$ in $G$, a contradiction. 
	
In summary, the following value is a more precise upper bound on the total weight of the vertices outside $P$:
\begin{align*}
&B-\sum_{v\in V_{1,2}\cup V_{3}^{D}}2D\rho(v)-\sum_{v\in V_{3}^{D-1}}(2D-1)\rho(v)-2Dn_{3}^{D}\\
\leq{}& (n-D-1)D(n-D)+n_3^{D-1}(2D-n-1)-Dn_3^{D}-2Dr^*+D\min\{1,\rho^*\}\\
&\quad-\sum_{v\in V_{1,2}\cup V_{3}^{D}\cup V_{3}^{D-1}}(2D-1)\rho(v).
\end{align*}
Let us now consider the vertices on $P$. We have $W(u_0)=D(1+o_0)$, $W(u_D)=D(1+o_D)$, and $W(u_i)=\ecc(u_i)(2+o_i)$ for $i=1,\ldots,D-1$. Since $\ecc(u_i)=\delta_i+r_i$, the total weight of the vertices on $P$ is
\begin{align*}
& 2D+D(o_0+o_D)+\sum_{i=1}^{D-1}(\delta_i+r_i)(2+o_i)\\
={}& 2\sum_{i=0}^{D-1}\delta_i+2\sum_{i=1}^{D-1}r_i+\sum_{i=1}^{D-1}r_io_i+\sum_{i=0}^{D}\delta_io_i.
\end{align*}
Each edge that links a vertex $v$ outside $P$ to a vertex $u_i$ in $P$ contributes for $r_i\leq \rho(v)$ in the sum $\sum_{i=1}^{D-1}r_io_i$. Hence,
\[
\sum_{i=1}^{D-1}r_io_i \leq \sum_{v\in V_{1,2}}2\rho(v) +
\sum_{v\in  V_{3}^{D}\cup V_{3}^{D-1}}3\rho(v) \leq
\sum_{v\in V_{1,2}\cup V_{3}^{D}\cup V_{3}^{D}}3\rho(v).
\]
Since $2\sum_{i=1}^{D-1}r_i\leq 2r^*(D-1)$, we get the following valid upper bound on the total weight of the vertices on $P$:
\[
2\sum_{i=0}^{D-1}\delta_i+ \sum_{i=0}^{D}\delta_io_i +2r^*(D-1) +
\sum_{v\in V_{1,2}\cup V_{3}^{D}\cup V_{3}^{D}}3\rho(v).  \]
Summing up the bounds for the vertices outside $P$ with those on $P$, we get the following upper bound for the total weight of the vertices in $G$:
\begin{multline*}
(n-D-1)D(n-D)+n_3^{D-1}(2D-n-1)-Dn_3^{D} +
2\sum_{i=0}^{D-1}\delta_i + \sum_{i=0}^{D}\delta_io_i \\
-\sum_{v\in V_{1,2}\cup V_{3}^{D}\cup V_{3}^{D-1}}(2D-4)\rho(v) -
2r^*+D\min\{1,\rho^*\}.
\end{multline*}
Let us decompose this bound into two parts $A_1+A_2$ with $A_1$ being equal to the sum of the first terms of the above upper bound, and $A_2$ being equal to the sum of the last ones:
\begin{align*}
A_1 &= (n-D-1)D(n-D)+n_3^{D-1}(2D-n-1)-Dn_3^{D} +
  2\sum_{i=0}^{D-1}\delta_i+ \sum_{i=0}^{D}\delta_io_i\\
A_2 &= -\sum_{v\in V_{1,2}\cup V_{3}^{D}\cup V_{3}^{D-1}}(2D-4)\rho(v)
  -2r^*+D\min\{1,\rho^*\}.
\end{align*}
\begin{shortlist}
  \item If $r^*=0$, then $A_2=0$, which implies $A_1+A_2=A_1$.
  \item If $\rho^*>0$, then $A_2\leq 4-2D-2r^*+D=4-D-2r^*<0$, which implies
    $A_1+A_2<A_1$.
  \item If $r^*>0$ and $\rho^*=0$, then $A_2=-2r^*<0$, , which implies
    $A_1+A_2<A_1$.
\end{shortlist}

In summary, the best possible upper bound is $A_1$ and is attained only if $r^*=0$, $n_0=0$, $\ecc(v)=D$ for all vertices in $V_{1,2}$, and the vertices outside $P$ are pairwise adjacent. We now have to compare $A_1$ with $f(n,D)$. 

Let us start with $D=3$. In that case, we have
$f(n,3)=14+(n-4)(3n-4+\max\{0,7-n\})$, while $A_1=(n-4)3(n-3) + n_3^{2}(5-n)
- 3n_3^{3} +14 + \sum_{i=0}^{3}\delta_io_i$. Hence, the difference is :
\[
f(n,3)-A_1=(n-4)(5+\max\{0,7-n\})-n_3^{2}(5-n)+3n_3^{3}-\sum_{i=0}^{3}\delta_io_i.
\]
We have 
\[
\sum_{i=0}^{3}o_i\leq 3(n_3^{2}+n_3^{3})+2(n-4-n_3^{2}-n_3^{3})=2(n-4)+n_3^{2}+n_3^{3}.
\]
Since $o_0+o_3\leq n-4$ to avoid a path of length 2 joining $u_0$ to $u_3$, we have 
\[
\sum_{i=0}^{3}\delta_io_i\leq 3(n-4)+2(n-4+n_3^{2}+n_3^{3}). 
\]
Hence, 
\[
f(n,3)-A_1 \geq (n-4)\max\{0,7-n\}-n_3^{2}(7-n)+n_3^{3}.
\]
This difference is minimized if and only if $n_3^{3}=0$, while $n_3^{2}=0$ if $n>7$, $n_3^{2}=0,1,2$ or $3$ if $n=7$, and $n_3^{2}=n-4$ if $n<7$. In all such cases, we get $f(n,3)-A_1=0$. 
\begin{itemize}
  \item If $n=4$, there is no vertex outside $P$, and $G\simeq \extG{4}{3}{0}$
    which is the unique graph in $\mathcal{C}_4^{3}$.
  \item If $n=5$, $n_3^{2}=1$, which means that the unique vertex outside
    $P$ is adjacent to 3 consecutive vertices on $P$. Hence, $G\simeq
    \extG{5}{3}{1}$ which is the unique graph in $\mathcal{C}_5^{3}$.
  \item If $n=6$, $n_3^{2}=2$,  which means that both vertices outside $P$
    are adjacent to 3 consecutive vertices on $P$. If one of them is adjacent
    to $u_0,u_1,u_2$, while the other is adjacent to $u_1,u_2,u_3$, we have
    $G\simeq H_2$. Otherwise, we have $G\simeq \extG{6}{3}{2}$.
  \item If $n=7$, $n_3^{2}\in\{0,1,2,3\}$ and $n_{1,2}=3-n_3^{2}$. If
    $n_{1,2}>0$ then the vertices in $V_{1,2}$ are all adjacent to $u_0$ and
    $u_1$ or all to $u_2$ and $u_3$, since they are pairwise adjacent, and
    they all have eccentricity $3$. So assume without loss of generality, they
    are all adjacent to $u_0$ and $u_1$. Then the vertices in $V_{3}^{2}$ are
    all adjacent to $u_0,u_1,u_2$, else the vertices in $V_{1,2}$ would have
    eccentricity 2. But $G$ is then equal to $\extG{7}{3}{0},\extG{7}{3}{1}$
    or $\extG{7}{3}{2}$. If $n_{1,2}=0$, then the three vertices outside
    $P$ are all adjacent to three consecutive vertices on $P$. If they are
    all adjacent to $u_0,u_1,u_2$, or all to $u_1,u_2,u_3$, then $G\simeq
    \extG{7}{3}{3}$, else $G\simeq H_3$.
  \item If $n>7$, all vertices outside $P$ are adjacent to $u_0,u_1$, or
    to $u_2,u_3$ (so that they all have eccentricity 3). Hence, $G\simeq
    \extG{n}{3}{0}$.
\end{itemize}

Assume now $D\geq 4$. We have
\[
f(n,D)=2\sum_{i=0}^{D-1}\delta_i+\Big(n-D-1\Big)\Big( 2D-1+D(n-D)+\max\{0,3D-n-3\}\Big)
\] and
\[
A_1 = 2\sum_{i=0}^{D-1}\delta_i + (n-D-1)D(n-D) + n_3^{D-1}(2D-n-1) - Dn_3^{D}
+ \sum_{i=0}^{D}\delta_i o_i.
\]
Hence, the difference is:
\[
f(n,D)-A_1 = (n-D-1)(2D-1+\max\{0,3D-n-3\}) - n_3^{D-1}(2D-n-1) + Dn_3^{D}
- \sum_{i=0}^{D}\delta_io_i.
\]
We have
\[
\sum_{i=0}^{D}o_i\leq 3(n_3^{D-1} + n_3^{D}) + 2(n-D-1-n_3^{D-1}-n_3^{D})
= 2(n-D-1) + n_3^{D-1} + n_3^{D}.
\]
Let $p$ be the number of vertices linked to both $u_1$ and $u_{D-1}$. 
\begin{itemize}
  \item If $D\geq 5$, then $p=0$, else $d(u_0,u_D)\leq 4 <D$.
  \item If $D=4$, then no vertex outside $P$ linked to $u_1$ and $u_{D-1}$ can also be linked to $u_0$ or to $u_D$ since $d(u_0,u_D)$ would be strictly smaller than $4$. Since no vertex outside $P$ can be linked to both $u_0$ and $u_D$ (else $d(u_0,u_D)<3$) we have $o_0+o_D\leq n-D-1-p$ and $o_1+o_{D-1}\leq n-D-1+p$. Hence, $o_2\leq n_3^{D-1}+n_3^{D}$. So,
  \begin{align*}
    \sum_{i=0}^{D}\delta_io_i &\leq
      D(n-D-1-p)+(D-1)(n-D-1+p)+(D-2)(n_3^{D-1}+n_3^{D})\\
    &= (n-D-1)(2D-1)+(D-2)(n_3^{D-1}+n_3^{D})-p.
  \end{align*}
  This value is maximized for $p=0$.
\end{itemize}
Hence, in all cases, we have 
\[
\sum_{i=0}^{D}\delta_io_i\leq (n-D-1)(2D-1)+(D-2)(n_3^{D-1}+n_3^{D}).
\]
Hence, 
\[
f(n,D)-A_1\geq (n-D-1)\max\{0,3D-n-3\}-n_3^{D-1}(3D-n-3)+2n_3^{D}.
\]
This difference is minimized if and only if $n_3^{D}=0$, while $n_3^{D-1}=0$ if $n>3(D-1)$, $n_3^{D-1}\in\{0,\ldots,n-D-1\}$ if $n=3(D-1)$, and $n_3^{D-1}=n-D-1$ if $n<3(D-1)$. In all such cases, we get $f(n,D)-A_1=0$. 
\begin{itemize}
  \item If $n<3(D-1)$, then all vertices outside $P$ are adjacent to 3 consecutive vertices on $P$. They are all adjacent to $u_0,u_1,u_2$, or all adjacent to $u_{D-2},u_{D-1},u_D$, else $d(u_0,u_D)\leq 3<D$. Hence, we have $G\simeq \extG{n}{D}{n-D-1}$.
  \item If  $n=3(D-1)$, $n_3^{D-1}\in\{0,\ldots,n-D-1\}$ and $n_{1,2}=2D-2-n_3^{D-1}$. If $n_{1,2}>0$ then the vertices in $V_{1,2}$ are all adjacent to $u_0$ and $u_1$ or all to $u_{D-1}$ and $u_{D}$, since they are pairwise adjacent, and they all have eccentricity $D$. So assume without loss of generality, they are all adjacent to $u_0$ and $u_1$. Then the vertices in $V_{3}^{D-1}$ are all adjacent to $u_0,u_1,u_2$, else $d(u_0,u_D)\leq 3<D$. But $G$ is then equal to $\extG{n}{D}{n_3^D}$. If $n_{1,2}=0$, then all vertices outside $P$ are adjacent to $u_0,u_1,u_2$, or all of them are adjacent to $u_{D-2},u_{D-1},u_D$, else $d(u_0,u_D)\leq 3<D$. Hence, $G\simeq \extG{n}{D}{n-D-1}$.
  \item If $n>3(D-1)$, all vertices outside $P$ are adjacent to $u_0,u_1$, or to $u_2,u_3$ (so that they all have eccentricity D). Hence, $G\simeq \extG{n}{D}{0}$.
\end{itemize}
\end{proof}

\section{Results for a fixed order and no fixed diameter}
	
We now determine the connected graphs that maximize the eccentric connectivity index when the order $n$ of the graph is given, while there is no fixed diameter. Clearly, $\K{3}$ and $\Path{3}$ are the only connected graphs of order $n=3$ and $\eci(\K{3})=\eci(\Path{3})=6$. For $n>3$, $\eci(\M{n})=2n^2-4n-2(n\bmod 2) >n^2-n =\eci(\K{n})$, which means that the optimal diameter is not $D=1$.
\begin{itemize}
  \item If $n=4$, $f(4,3)=14<\eci(\M{4})=16$, which means that $\M{4}$ has maximum eccentric connectivity among all connected graphs with 4 vertices.
  \item If $n=5$, $f(5,3)=27$, $f(5,4)=24$ and $\eci(\M{5})=30$, which means that $\M{5}$ and $H_1$ have maximum eccentric connectivity index among all connected graphs with 5 vertices.
  \item If $n=6$, $f(6,3)=44$, $f(6,4)=42$, $f(6,5)=38$ and $\eci(\M{6})=48$, which means that $\M{6}$ has maximum eccentric connectivity index among all connected graphs with 6 vertices.
\end{itemize}
Assume now $n\geq 7$. We first show that lollipops are not optimal. Indeed, consider a lollipop $\extG{n}{D}{0}$ of order $n$ and diameter $D$. 
\begin{itemize}
  \item If $D=n-1$, then $G\simeq \Path{n}$ which implies 
    \begin{align*}
      \eci(\extG{n}{n-1}{0}) &= \sum_{i=1}^{D-1}2\max\{i,D-i\}+2D =
        \frac{3D^2+D\bmod 2}{2}\\
    &\leq \frac{3D^2+1}{2} = \frac{3n^2}{2}-3n+2 < 2n^2-4n-2 \leq \eci(\M{n}).
    \end{align*}
  \item If $D<n-1$ then either $n<3(D-1)$, and we know from Corollary \ref{cor:max_k} that $\eci(\extG{n}{D}{n-D-1})>\eci(\extG{n}{D}{0})$, or $n\geq 3(D-1)$, in which case we show that  $\eci(\extG{n}{D+1}{n-D-2})>\eci(\extG{n}{D}{0})$. Since $2\sum_{i=0}^{D-1}\max\{i,D-i\}=\frac{3D^2 + D \bmod 2}{2}$, we know from Lemma \ref{lem:eci_for_extremal_graph} that
\begin{align*}
\eci(\extG{n}{D+1}{n-D-2})={}
   &2\sum_{i=0}^{D}\max\{i,D+1-i\}\\
   &+\Big(n-D-2\Big)\Big( 2(D+1)-1+(D+1)(n-D-1)\Big)\\
   &+\Big(n-D-2\Big)\Big(2(D+1)-n-1+(D+1)-2\Big)\\
={}&\frac{3(D+1)^2 + (D+1) \bmod 2}{2}+\Big(n-D-2\Big)\Big(3D+D(n-D)\Big)
\end{align*}
and
\begin{align*}
\eci(\extG{n}{D}{0}) &=
  2\sum_{i=0}^{D-1}\max\{i,D-i\}+\Big(n-D-1\Big)\Big( 2D-1+D(n-D)\Big) \\
&=\frac{3D^2 + D\bmod 2}{2}+\Big(n-D-1\Big)\Big( 2D-1+D(n-D)\Big).
\end{align*}
Simple calculations lead to
\[
\eci(\extG{n}{D+1}{n-D-2})-\eci(\extG{n}{D}{0})=n-2D+(D-1) \bmod 2\geq n-2\left(  \frac{n}{3}+1 \right)=\frac{n}{3}-2>0.
\]
\end{itemize}
Hence, the remaining candidates to maximize the eccentric connectivity index when $n\geq 7$ are $\M{n}$ and $\extG{n}{D}{n-D-1}$. Let
\[
g(n)=\max_{D=\left\lceil\frac{n}{3}+2\right\rceil}^{n-D-1}
\eci(\extG{n}{D}{n-D-1}).\]
We can rewrite $\eci(\extG{n}{D}{n-D-1})$ as follows:
\[
\eci(\extG{n}{D}{n-D-1})=D^3-D^2(n+\frac{5}{2})+D(n^2+5n-1)-n^2-3n+4+ D\bmod 2.
\]
It is then not difficult to show that $g(n)=\eci(\extG{n}{D^*}{n-D^*-1})$ with $D^*=\lceil\frac{n+1}{3}\rceil+1$, and simple calculations lead to
\[
g(n)=\frac{1}{54}(8n^3+21n^2-36n+
\begin{cases}
0&\text{if }n \bmod 6=0\\
6n+1&\text{if }n \bmod 6=1\\
32&\text{if }n \bmod 6=2\\
27&\text{if }n \bmod 6=3\\
6n+28&\text{if }n \bmod 6=4\\
59&\text{if }n \bmod 6=5\\
\end{cases}\quad).
\]
We then have $g(7)=66<68=\eci(\M{7})$, which means that $\M{7}$ has the largest eccentric connectivity among all graphs with $7$ vertices. Also, $g(8)=96=\eci(\M{8})$, which means that both $\extG{8}{4}{3}$ and $\M{8}$ have the largest eccentric connectivity index among all graphs with $8$ vertices. For graphs of order $n\geq 9$, we have $\frac{8n^3+21n^2-36n}{54} > 2n^2-4n$, which means that $\extG{n}{D^*}{n-D^*-1}$ is the unique graph with largest eccentric connectivity index among all graphs with $n$ vertices.
These results are summarized in Table \ref{table1}, where ${\eci_n}^*$ stands for the largest eccentric connectivity index among all graphs with $n$ vertices.

\begin{table}[h!]
	\begin{center}
		\caption{Largest eccentric connectivity index for a fixed order $n$}
		\label{table1}
		
		\begin{tabular}{ccc}
			$n$ & ${\eci_n}^*$ & optimal graphs \\
			\hline
			3&6&$\K{3}$ and $\Path{3}$\\
			4&16&$\M{4}$\\
			5&30&$\M{5}$ and $H_1$\\
			6&48&$\M{6}$\\
			7&68&$\M{7}$\\
			8&96&$\M{8}$ and $\extG{8}{4}{3}$\\
			$\geq 9$&$g(n)$&$\extG{n}{\left\lceil\frac{n+1}{3}\right\rceil+1}{n-\left\lceil\frac{n+1}{3}\right\rceil-2}$
		\end{tabular}
	\end{center}
\end{table}
Note finally that Tavakoli~\etal~\cite{Tavakoli2014} state that 
$g(n)=\eci(\extG{n}{D}{n-D-1})$ with $D=\lceil\frac{n}{3}\rceil+1$ while we have shown that the best diameter for a given $n$ is $D=\lceil\frac{n+1}{3}\rceil+1$. Hence for all $n\geq 9$ with $n\bmod 3=0$, we get a better result. For example, for $n=9$, they consider $\extG{9}{4}{4}$ which has an eccentric connectivity index equal to 132 while $g(9)$=134.

\section{Conclusion}
We have characterized the graphs with largest eccentric connectivity index among those of fixed order $n$ and fixed or non-fixed diameter $D$. It would also be interesting to get such a characterization for graphs with a given order $n$ and a given size $m$. We propose the following conjecture which is more precise than the one proposed in~\cite{Zhang2014}
\begin{conj*}
	
  Let $n$ and $m$ be two integers such that $n\geq 4$ and $m\leq{ {n-1}\choose{2}}.$ Also, let
	\[D=\left\lfloor\frac{2n+1-\sqrt{17+8(m-n)}}{2}\right\rfloor \text{ and } k=m-{{n-D+1}\choose{2}}-D+1
	\]
	Then, the largest eccentric connectivity index among all graphs of order $n$ and size $m$ is attained with $\extG{n}{D}{k}$. Moreover, 
	\begin{itemize}
		\item if $D>3$ then $\eci(G)<\eci(\extG{n}{D}{k})$ for all other graphs $G$ of order $n$ and size $m$.
		\item if $D=3$ and $k=n-4$, then the only other graphs $G$ with $\eci(G)=\eci(\extG{n}{D}{k})$ are those obtained by considering a path $u_0-u_1-u_2-u_3$, and by joining $1\leq i\leq n-3$ vertices of a clique $\K{n-4}$ to $u_0,u_1,u_2$ and the $n-4-i$ other vertices of $\K{n-4}$ to $u_1,u_2,u_3$.
	\end{itemize}
\end{conj*}

\bibliographystyle{acm}

\end{document}